\documentclass[sigconf,9pt, nonacm]{acmart}

\usepackage{amsmath}

\newif\iffinalcopy \finalcopyfalse     %

\newcommand{\figorbox}[3]{%
  \IfFileExists{#1}{\includegraphics[width=#2]{#1}}%
  {\fbox{\parbox[c][#3][c]{\dimexpr#2-2\fboxsep-2\fboxrule}{\centering\small\texttt{\detokenize{#1}}}}}}

\begin{document}

\title{Cambium: Participant-Built Commitment Roots for Shared Digital Infrastructure}

\author{Mansoor ``Manny'' Ahmed}
\affiliation{%
  \institution{OpenOrigins Inc}
  \city{San Francisco}
  \country{USA}}
\affiliation{%
  \institution{University of Cambridge}
  \city{Cambridge}
  \country{United Kingdom}}
\email{m@openorigins.com}

\author{Dann Toliver}
\affiliation{%
  \institution{TODAQ}
  \city{Toronto}
  \country{Canada}}
  \affiliation{%
  \institution{University of Cambridge}
  \city{Cambridge}
  \country{United Kingdom}}
\email{dann.toliver@todaq.net}

\author{Ekaterina Pavlova}
\affiliation{%
  \institution{Independent researcher}
  \city{London}
  \country{United Kingdom}}
\email{ktpavlva@gmail.com}

\renewcommand{\shortauthors}{Ahmed et al.}

\begin{abstract}
Commitment services, such as timestamping, transparency logs and provenance
registries, let a party prove that it bound itself to a record without
publishing the record. Today an operator runs them, or a blockchain replaces the
operator at the cost of network-wide agreement on every write. We present
Cambium, a permissioned peer-to-peer service in which participants build each
cycle's commitment root together, so that no single party decides which inputs
the root includes. Fixed-size committees certify inputs and use a
recursive-doubling merge, in which disjoint groups recompute every
intermediate value, to spread root construction across all participants. A
successful cycle costs each node $O(c\log(n/c))$ messages for committee size
$c$, and each node stores its own $O(\log(n/c))$ proof per cycle, along with the
$O(n)$ membership roster. A Go prototype on AWS with $c=4$, 8--256 nodes and
simulated stragglers completes cycles in 12--148\,ms and verifies proofs in
under 1.5\,ms, and demonstrates logarithmic growth.
We show that
conflicting certificates identify their signers, that a proof binds a signed
input to its root, and that two conflicting roots both become final only if at
least a third of the nodes are Byzantine, except with probability exponentially
small in the number of committees. We also
set out the membership, proof-availability and exclusion questions a community
deploying Cambium must answer.
\end{abstract}



\maketitle

\section{Introduction}

Many digital interactions need a commitment rather than a ledger: evidence that a
party bound itself to a record at a given point, which it cannot later disown and
which others can check. A cooperative registering dataset contributions, a
newsroom committing to a photograph's hash at capture, and an agent committing to
an offer before a counterparty acts all need that evidence across organisations.
None needs a total order over every action. Nor do the records need to live in
one place: local-first software~\cite{kleppmann2019localfirst} and per-agent
ledgers such as TrustChain~\cite{otte2020trustchain} keep records with their
participants, who disclose a record and its proof only when needed.

Timestamping and transparency logs provide compact inclusion
proofs~\cite{haber1991timestamp,rfc9162}, but an operator chooses the root's
contents. Blockchains distribute that control at the cost of network-wide
agreement on each write.

Cambium sits between the two. All participants build one root per cycle
together: no node proposes it or receives every input, and each keeps only its
own proof. Cambium does not aim to beat a single aggregator on latency or cost;
it aims to remove the aggregator's control over which inputs enter the root.
Omitting a committee's inputs needs a quorum of that committee or shows up as a
counted hole, while each node sends only $O(c\log(n/c))$ messages per cycle
(\S\ref{sec:protocol}). We establish local certificate uniqueness with
accountability, root-relative proof soundness and finality: conflicting final
roots mean
that at least a third of the nodes are Byzantine, except with probability
exponentially small in $n/c$ (\S\ref{sec:security}). Liveness is weaker:
withholding below that bound can stall cycles. Equal protocol roles also do not settle who may join,
whose inputs are lost when links fail, or who restarts a stalled network
(\S\ref{sec:operating}).

We contribute a commitment-service interface that separates record custody,
proof verification and root acceptance, illustrated by a use case
(\S\ref{sec:req}); the cycle trie and banyan trie (\S\ref{sec:design}); the cycle
protocol and a security analysis of accountability, proof soundness, finality
and availability (\S\ref{sec:protocol}--\ref{sec:security});
a Go prototype and a crash simulation (\S\ref{sec:eval}); and the requirements for operating Cambium as shared
infrastructure (\S\ref{sec:operating}).

\section{Overview and Scope}\label{sec:req}

We write \emph{node} for a participating computer, $n$ for the number of nodes
and $c$ for the committee size. Every node holds the \emph{roster}, the list of
all members' identities and public keys, which the genesis root $CR_0$ fixes and
only signed membership epochs change (\S\ref{sec:security}, \S\ref{sec:operating}).

In each cycle every node submits one input, a
digest, and committees of $c$ nodes certify these inputs (\S\ref{sec:protocol})
before all nodes jointly compute a cycle root over them. Each root commits to
its predecessor, so the roots a node accepts form a hash chain. We want five
properties: conflicting certificates identify their signers (accountability); a
proof verifies against a root only if its input was certified for that root
(proof soundness); an accepted root omits at most a fixed share of the
committees, set by $H_{\max}$, and counts each omission as a hole
(\S\ref{sec:protocol});
honest nodes never hold conflicting final roots (finality); and a node's
per-cycle cost does not grow with the length of the chain. Section~\ref{sec:security} establishes the first two,
and the fourth, while fewer than a third of nodes are Byzantine, with high
probability; the acceptance rule and cost
bounds of \S\ref{sec:protocol} give the third and fifth. A node exposes four
operations: submit an input for a future cycle, export a proof, verify a proof
against a root, and report the root it accepts for a cycle and whether that root
is final.

\textbf{Use case.} Consider news organisations that commit to photographs' hashes
at capture, so that a fact-checker can later confirm that a photograph existed,
unaltered, by a given cycle. The participants are hundreds to thousands of
organisations, each running one node that batches its photographers' captures.
They need commitment within about a minute, any number of commitments per cycle,
tolerance of members going offline, and rare roster changes. With $c=16$ Cambium
tolerates 5\% of members offline, and with $c=32$, 10\% (Table~\ref{tab:sim}); more
slows progress but does not affect safety. A newsroom node
submits the root of its local Merkle batch (\texttt{POST /commit}); at publication it exports a proof for
that cycle with the photograph's batch-inclusion claim (\texttt{GET /path}). A
fact-checker checks the committee signatures, recomputes the cycle root
(\texttt{POST /verifyProof}), and compares it with the final root that most of
several randomly chosen members report (\texttt{GET /root}).

\textbf{What a proof establishes.} A proof binds a signed digest to a cycle root
and a key, without attesting to the record's truth or the key's authority. The
cycle does not establish wall-clock time, and guessable
records need a salted digest $H(r\Vert x)$ with fresh $r$. After rollback
(\S\ref{sec:protocol}), proofs against a discarded root no longer verify, and
holders must re-submit. A proof against a root outside the chain the verifier
trusts provides no evidence about the claim.

\textbf{Batching and absence.} A complete committee packet shows that a digest
is absent from that packet. A claim that a node made no conflicting commitment
over a range of cycles additionally needs a namespace, a rule for unique
submissions, and the complete batches for those cycles. A hole
(\S\ref{sec:protocol}) establishes neither an empty input nor the absence of an
action. There is no globally readable ledger, and each node keeps the $O(n)$
roster, this storage cost is the main non-logarithmic growth factor in Cambium.

\section{Data Structures}\label{sec:design}

\subsection{Cycle trie}\label{sec:cycle}

A Merkle tree over the committees' packets yields a root, but if a few nodes
compute its upper levels they control it, and if every node computes every level,
every node receives every packet. Cambium instead follows recursive-doubling
all-reduce~\cite{thakur2005mpich}: at each level, distinct values halve as their
holders double, so after $\log_2(n/c)$ levels every node holds the root. To
tolerate faulty nodes, nodes act in groups of $c$, and a node accepts a value
only when a quorum of a group vouches for it. The group assignment below ensures
that every level-$l$ value is recomputed by $2^l$ disjoint groups, so a captured
group causes divergence that later cycles can expose (\S\ref{sec:protocol}); that
each node contacts $O(c)$ others per level; and that a node's co-members change
across levels, decorrelating the recomputations (Fig.~\ref{fig:cycle}).

\begin{figure*}[t]
\centering
\begin{minipage}[c]{0.28\textwidth}
\centering
\includegraphics[width=\linewidth]{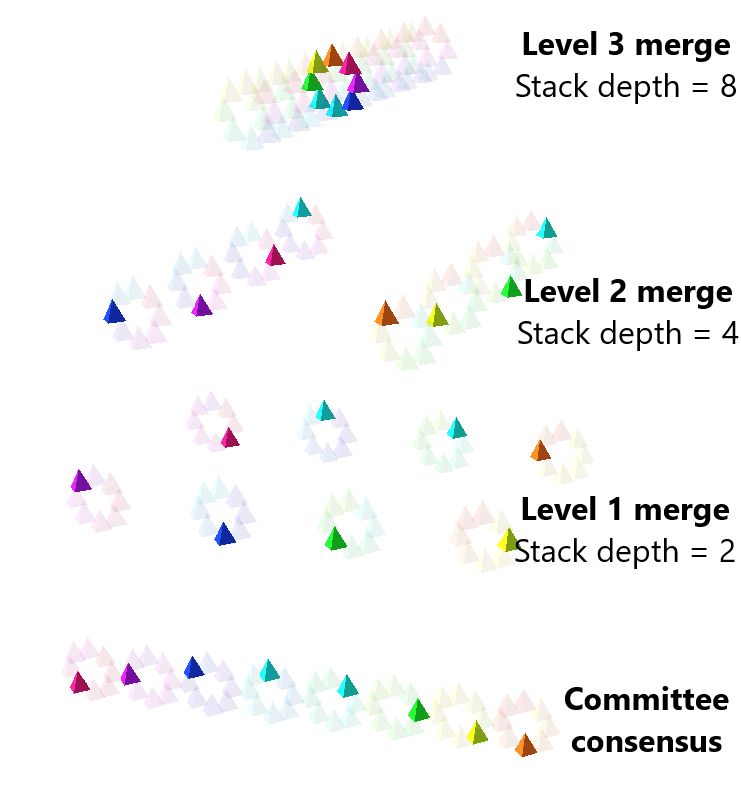}
\end{minipage}\hfill
\begin{minipage}[c]{0.68\textwidth}
\centering
\includegraphics[width=\linewidth]{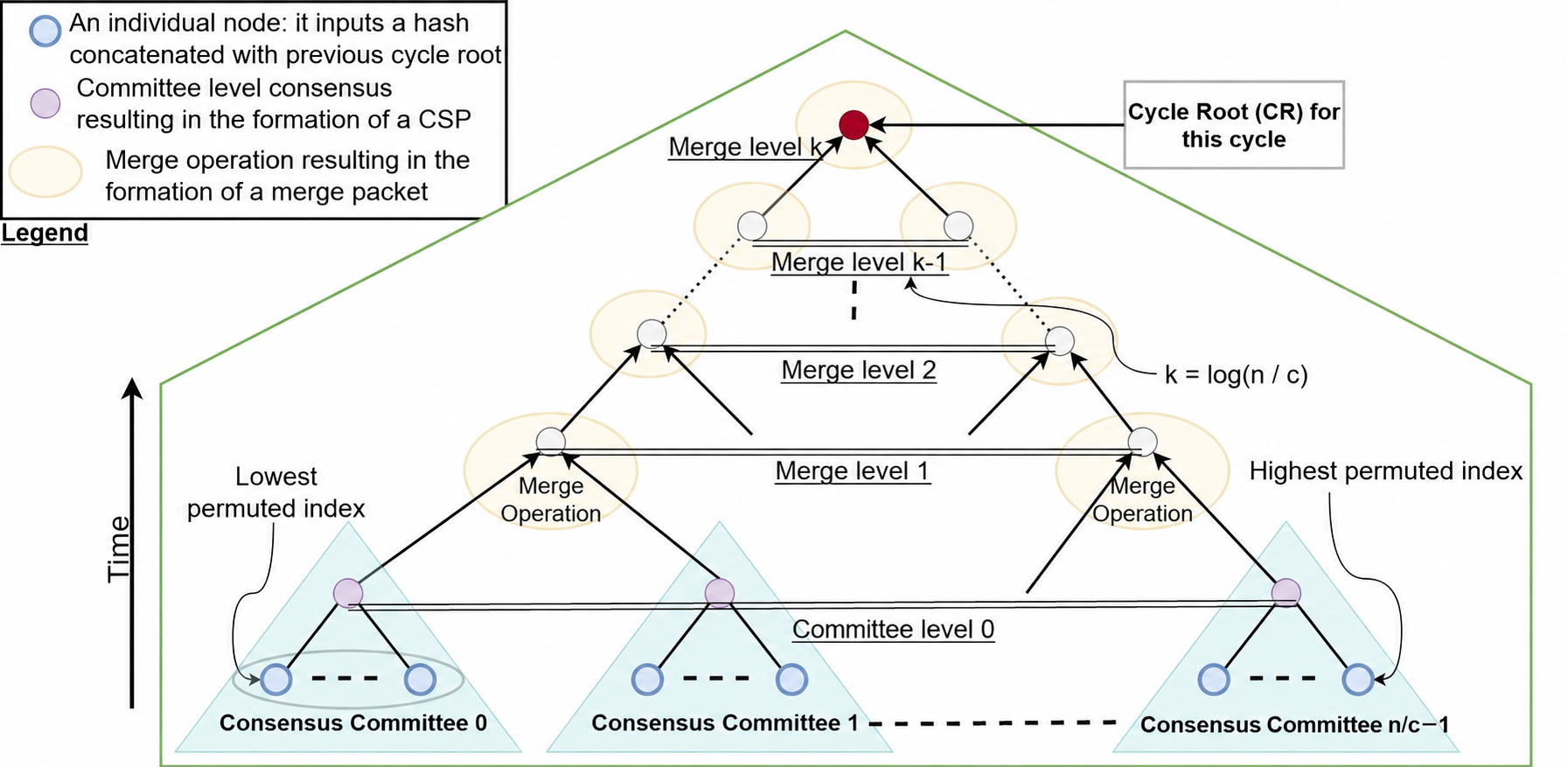}
\end{minipage}
\caption{Group layout and cycle-root construction. Left: an illustrative cycle
with $n=64$, $c=8$ and three merge levels. Right: signed committee packets merge
into $CR_t$; the image's $k$ is $L$.}
\Description{Left: four rows depict committees and three merge levels with
stack depth doubling. Right: committee packets merge through successive levels
into one cycle root.}
\label{fig:cycle}
\end{figure*}

Let $n=c\cdot 2^L$ with $c$ a power of two; other roster sizes need padding or
unequal committees, which we do not analyse. Each node has a cycle index $i$
(\S\ref{sec:banyan}). At level $l\in\{0,\dots,L\}$, where level 0 is the
committee and merges occur at levels 1 to $L$, the low $l$ bits of $i$
give its default group $q = i \bmod 2^l$, the next $\kappa=\log_2 c$ bits give
its position $j = \lfloor i/2^l\rfloor \bmod c$, and the high bits give its
stack. Its group is
\begin{equation}\label{eq:group}
\mathrm{group}(i,l) = (q - o_l[j]) \bmod 2^l,
\end{equation}
where $o_l$ is a fixed table of integer offsets. A level-$l$ stack is the
$c\cdot 2^l$ nodes that share the high bits; it holds $2^l$ groups of $c$ nodes,
one per position, and every group in a stack computes the same value. At level
0, $q=0$ and each stack is one committee of $c$ consecutive indices. Moving up a
level shifts the position field up one bit, halving the number of stacks and
doubling their depth. For example, with $n=16$, $c=4$ and the prototype's
diagonal offsets $o_l[j]=j\bmod 2^l$, node 6 is in committee $\{4,\dots,7\}$ at
level 0, in group $\{1,2,5,6\}$ of stack $\{0,\dots,7\}$ at level 1, and in group
$\{1,6,11,12\}$ at level 2.

The offsets permute membership within a stack, so root and proof validity do not
depend on $o_l$; agreement does, since $o_l$ sets how often two nodes share a
group across levels. Choosing $o_l$ is a combinatorial design problem we treat
separately.

\subsection{Banyan trie}\label{sec:banyan}

After each cycle, a small-domain pseudorandom permutation~\cite{Stefanov2012FastPRPFP},
keyed by the locally accepted root, maps identity indices to next-cycle indices.
This re-seats committees; assignment is unknowable until the preceding root
exists. Every packet carries its predecessor root,
linking successive cycle tries into the \emph{banyan trie}. Re-seating also samples
global health: nodes that built on different predecessor roots can meet in the
next cycle's committees, where a node can join, as a dissenter, a committee that
a foreign root's assignment places it in (\S\ref{sec:protocol}).

\section{Cycle Protocol}\label{sec:protocol}

\textbf{Committee phase.} A committee must turn its members' inputs into one
packet signed by $\tau=\lceil 2c/3\rceil$ members, so that its certificate is
unique (Lemma~\ref{lem:attr}). Members first exchange signed initial packets
carrying the cycle number, sender identity, cycle index, predecessor root and a
hash of the sender's input. Two rounds follow. In the \emph{hash round} members
agree on which input hashes the packet includes; this fixes the inputs before
anyone reveals them, so no member can choose its input after seeing others'. In
the \emph{data round} they reveal the inputs and agree on the committee signed
packet (CSP), the leaf of the cycle trie. Each round runs a single-slot
PBFT-style agreement~\cite{castro1999pbft} with bounded view changes; a
commit certificate contains $\tau$ member signatures. Without view changes, its
$O(c^2)$ messages cost each node $O(c)$; committee designs without agreement
are future work. A committee that
cannot decide within the view-change bound contributes a hole of
weight $c$, its keys; a hole is a fixed string, so all groups hash it alike.

\textbf{Merge phase.} At each merge level, a node combines the value of its own
side, which it already holds, with the value of the sibling side, which it must
learn from nodes that hold it. Level-$l$ stack $s$ is formed from
level-$(l{-}1)$ stacks $2s$ and $2s{+}1$. At merge level $l$, a node's
\emph{peers} are the other members of its level-$(l{-}1)$ group, which hold its
side's value, and its \emph{allies} are the group with the same group number in
the sibling level-$(l{-}1)$ stack, which hold the other side's value; the
relation is symmetric, and both sets follow from Eq.~\eqref{eq:group} and the
roster. In the example of \S\ref{sec:cycle}, node 6's allies at level 2 are $\{9,10,13,14\}$,
one of the two groups that compute the value of nodes 8--15. Nodes exchange
signed merge packets with peers and allies. A node accepts the ally value
carried by $\tau$ distinct allies, discarding any ally seen sending two values,
and otherwise records a hole. With $MP_0$ the CSP, it then computes, for
$l\ge1$,
\[
MP_l = H\!\left(MP^{\mathrm{left}}_{l-1} \,\Vert\, MP^{\mathrm{right}}_{l-1} \,\Vert\, CR \,\Vert\, HC \,\Vert\, DC \,\Vert\, l\right),
\]
where $CR$ is the predecessor root, $HC$ the hole count and $DC$ the dissent
count. A hole at level $l\ge1$ adds $c\,2^{l-1}$ to $HC$, the number of keys
it hides, so $HC\le n$. Merges sum $HC$ and $DC$ over disjoint subtrees.

Merge packets show their work. With a window of $w$ levels, a node attaches to
each merge packet the $w$ levels of merges beneath the value it vouches for, with
their signatures, so that a receiver can recompute that value, counts included, from certified
inputs; alternatively, or as well, it attaches its own merge path down to its
committee. A group with a Byzantine quorum, $\tau$ Byzantine members, therefore
cannot invent a value. It
can record a hole for a value it received, which $HC$ counts but which splits its
allies from the rest; omitting inputs without a hole also needs captured groups
at the levels whose work it shows. With $w=0$ and no path there is no such
check, so an instance must use at least one. The same evidence repairs holes.
Some of a node's peers at the next level computed the same value through other
ally groups, so a node that recorded a hole can fill it from a peer's window that
carries the missing value, and recompute.

Apart from filling holes from windows, nodes do not reconcile values with peers,
so the $2^l$ groups computing a level-$l$ value agree only if they end with the
same holes and dissent. A merge goes wrong when they do not: with
$c=4$, if two members of one group's ally group have crashed, that group records
a hole while the other groups computing the same value receive it, and unless a
window fills the hole, the two sets of nodes finish on different roots. Such holes can cascade (\S\ref{sec:eval}), and the next cycle
reports the split as dissent.

\textbf{Dissent.} Dissent detects a split network (\S\ref{sec:banyan}) so that
acceptance rolls back instead of continuing on divergent roots. A node that receives a committee
packet built on a foreign root $CR'_{t-1}$, from a sender whose committee under
$CR'_{t-1}$ would include it, joins that committee as a dissenter as well as
taking part in its own. It does so only if $CR'_{t-1}$ builds on its own
$CR_{t-2}$: cycle $t-1$'s assignment is then common to both roots, so the
foreign committee can check the dissenter's evidence. Dissent must be provable.
The dissenter attaches its BuildPath, its record of cycle $t-1$, from its committee
certificate up to $CR_{t-1}$, with the $\tau$ peer and ally signatures that each
group sends alongside the merge values it creates, the root included; nodes keep
these until a cycle is final. A CSP's dissent field
counts the dissenters it includes, and a committee hole contributes zero. A node
joins at most one committee per foreign root, so the merged sum $DC$ counts
dissenting keys without double counting. Holes can still erase dissent: two Byzantine members of a four-node committee
can withhold signatures, turning an honest dissenter's committee into a hole. An
accepted root's holes hide at most $H_{\max}$ keys, and dissenters land where the foreign
root's assignment places them, so holes can hide at most about one eighth of a
split's dissent.

\textbf{Acceptance, rebuild and rollback.} A node accepts $CR_t=MP_L$ if it
verifies, $DC \le D_{\max}$ and $HC+DC \le H_{\max}$: holes and dissenters never
count toward a quorum of $\tau$. Both thresholds count keys out of $n$. We use $H_{\max}=n/8$ and
$D_{\max}=n/16$ as placeholders, not derived thresholds; $H_{\max}$ tolerates no
committee hole when $n/c<8$.
$D_{\max}$ is a policy knob: a positive value tolerates small splits and leaves
the minority to catch up (\S\ref{sec:security}), while $D_{\max}=0$ rolls back on
any verified split. Since $H_{\max}=n/8$, an accepted root may omit up to one
eighth of the committees: the root
counts the omission but does not attribute it, and an adversary that places
$\lceil c/3\rceil$ members in a committee chooses whose inputs are omitted.
Otherwise, exceeding $H_{\max}$ triggers a rebuild of cycle $t$, a new attempt; since the permutation key is
$CR_{t-1}$, a rebuild re-derives the same assignment, so a committee blocked by
withholders stays blocked. Excess dissent means many nodes accepted a different
$CR_{t-1}$: nodes return to $CR_{t-2}$ and rebuild cycle $t-1$; dissent takes
precedence over holes. Finality bounds rollback, so proofs stay usable: once a
node accepts cycle $t$, it treats $CR_{t-2}$ as final and never rolls back past
it. Two roots for one cycle that each gain an accepted successor, like
conflicting final roots, need more Byzantine nodes than
Theorem~\ref{thm:final} allows; nodes that see either halt, keep the evidence,
and restart only as the governance of \S\ref{sec:operating} decides.

\textbf{Proofs and root trust.} For each cycle, a node retains an internal
BuildPath containing its committee result and merge records, from which it
exports a proof: the committee certificate and one adjacent subtree summary per
merge level, with an optional local batch-inclusion claim. A verifier checks
that the committee signers form the committee assigned under $CR_{t-1}$,
recomputes the cycle root, and
compares it with a root it accepts for that cycle. The exported proof carries no
peer or ally signatures and does not certify network-wide acceptance of the root.
An outside verifier therefore needs an independent rule for choosing trusted
roots (\S\ref{sec:security}); compact aggregate root
certificates~\cite{boneh2018compact} remain future work.
With one successful build, a node sends $O(c(1+L))$ messages, whether or not it
has anything to commit; proof size and verification work are $O(L)$ given the
roster and $CR_{t-1}$, and retaining BuildPaths or proofs for $E$ cycles costs
$O(EL)$ for fixed $c$. Windows and paths enlarge merge packets, by about
$2^w-1$ merges or $L$ levels of signatures, but add no messages. These bounds
exclude retries and dissent.

\section{Security}\label{sec:security}

\emph{Model.} Within each membership epoch, the roster and public keys are fixed;
a signed future epoch can change the roster at a specified cycle. Hashes are collision
resistant; signatures are unforgeable and bind cycle, attempt, predecessor root, level,
position, phase and value; honest nodes sign at most one value per context. A
static adversary controls $f$ of the $n$ identities, fixed before a cycle's
permutation is known; adaptive corruption is outside the model. Honest messages
arrive within a bound $\Delta$ from which timeouts are set; a message delayed
beyond $\Delta$ has the same local effect as a failure.

\begin{lemma}[Local certificate uniqueness]\label{lem:attr}
For a fixed group of $c$ identities and a fixed signing context, two certificates
with $\tau=\lceil 2c/3\rceil$ signers share at least $2\tau-c$ signers. If fewer
than $c/3$ members are Byzantine, two distinct non-hole values cannot both be
certified; any two conflicting certificates identify at least $2\tau-c$
double-signers.
\end{lemma}
\noindent Two signer sets of size $\tau$ in a group of $c$ intersect in at least
$2\tau-c\ge\lceil c/3\rceil$ members, so with fewer than $c/3$ Byzantine members the
intersection contains an honest one, who signs once per context. For $c=16$,
conflicts name six signers. This is the accountability Cambium offers, as in
accountable agreement~\cite{civit2021polygraph}: equivocation is attributable,
withholding is not. Nodes keep conflicting certificates as evidence against
their signers; what happens to those keys is a membership decision
(\S\ref{sec:operating}).

\emph{Root-relative proof soundness.} Given an authentic accepted root and
membership configuration, collision resistance and signature checks bind a
certified input to its key and cycle root without disclosing other records.

\emph{Group sampling.} A group can diverge once it holds $b=c-\tau+1$ Byzantine
members, which is $\lceil c/3\rceil$ for $c$ a power of two: by withholding, they
show a value to some honest nodes and a hole to others. Under uniform assignment without
replacement, a group of $c$ does so with probability
\begin{equation}\label{eq:pbad}
p_{\mathrm{bad}}=\sum_{k=b}^{c}\binom{f}{k}\binom{n-f}{c-k}\Big/\binom{n}{c}.
\end{equation}
Eq.~\eqref{eq:pbad} assumes an unbiased seed; withholding and forced rebuilds bias
it, so the figures below are optimistic,
and the static-adversary assumption covers corruption, not assignment. At
$n=16{,}384$, $f=5{,}461\approx n/3$ and $c=16$, $p_{\mathrm{bad}}=0.45$:
group-level divergence is common, and it costs liveness, not safety
(Theorem~\ref{thm:final}).

\begin{theorem}[Finality]\label{thm:final}
Suppose $f<n/3$, honest nodes sign only for the root they accept, and each
cycle's assignment is a uniform permutation. Two conflicting roots for one cycle
both become final with probability at most $G\,e^{-\delta n/c}$. Here $G$ is the
number of assignments the adversary can choose between, and
$\delta=D(\tfrac78\,\|\,q_c)$, where $D$ is binary relative entropy and $q_c$ is
the probability that $c$ nodes drawn from a two-thirds share of the roster
include $\tau$ of that share. For $c=4$ and $c=16$, $\delta\approx0.19$ and
$0.25$.
\end{theorem}
\begin{proof}[Proof sketch]
After a split at cycle $t$, each side's assignment for $t+1$ is keyed by its own
root, so each side must fill its own committees. Honest nodes sign for one root,
and holes and dissenters never count toward a quorum, so the sides' signer
shares, honest members plus all Byzantine ones, sum to at most $1+f/n<4/3$: all
but one side have under two thirds. Such a side accepts $t+1$ only if at least
seven eighths of its $n/c$ committees reach $\tau$ signers, while each does so
with probability below $q_c$; committee counts under a uniform permutation are
negatively associated, so a Chernoff bound gives $e^{-\delta n/c}$. A rebuild
re-derives the same assignment, so retries do not help, and a root becomes final
only after it gains accepted successors.
\end{proof}
\noindent At $n=16{,}384$ and $c=16$ the bound is $G\cdot10^{-111}$ per split; at
the evaluated $n=256$ and $c=4$ it is about $G\cdot10^{-5}$. A single deceived
node is a split of one: its root never gains a successor. The bound rests on
random assignment, as Eq.~\eqref{eq:pbad} does: adversarial placement could
finalise both sides from $f\ge0.2n$ at $c=16$. $G$ measures how far withholding
can steer a side's root before its assignment is fixed; seeding assignment from
an external beacon would make $G=1$. The proof needs only one accepted successor,
so finalising $CR_{t-1}$ should also be safe; we keep $CR_{t-2}$ to give dissent a
cycle to arrive, and leave the faster rule to future work.

\emph{Catch-up and root trust.} The stalled side of a split must rejoin the
final chain; techniques for moving from almost-everywhere to everywhere
agreement~\cite{king2009everywhere} suggest a periodic catch-up step. An outside
verifier need not
track every root: each root commits to its predecessor, so a verifier
that trusts one recent root, for instance one anchored in
Bitcoin~\cite{opentimestamps}, can check any earlier root on its chain from one
constant-size $MP_L$ preimage per intervening cycle.

\emph{Availability.} Safety holds for $f<n/3$; liveness does not, and denial of
progress is the residual attack. Dissent cannot be raised at will, because each
dissenter must present its own certified BuildPath. A key that claims a root no
honest node holds needs a false value certified somewhere on its path, and peer
signatures push this to the top level, where both the key's ally group and its
own top-level group, which signs the root, must hold a Byzantine quorum. Forcing
rollback needs more than $D_{\max}=n/16$ such keys: at $n=16{,}384$ and $c=16$,
about $f>0.59n$ in expectation, against $f>0.40n$ without peer signatures, so
below $n/3$ dissent cannot be forged; windows only raise these bounds. Withholding is stronger: committees with $b$ Byzantine members become holes at
will, and a cycle in which they exceed $H_{\max}$ stalls for as long as the
adversary withholds (\S\ref{sec:protocol}). Committee
holes alone reach $H_{\max}$ in expectation at $f/n\approx0.16$ for $c=4$ and
$0.22$ for $c=16$, and merge-level holes, each hiding up to $2^{L-1}$
committees, lower this further. Mixing the attempt number into the seed after a
failed rebuild would end the stall, at the cost of more influence over
assignment.

\section{Implementation and Evaluation}\label{sec:eval}

We implemented a Cambium prototype in Go.
We ran one node per AWS ECS Fargate task (2\,vCPU, 4\,GiB RAM) with $c=4$, an
EFS-backed proof store and the prototype's straggler simulation enabled.

\begin{table}[t]
\centering
\caption{Our AWS measurements, $c=4$, simulated stragglers enabled. Messages are
sent per node per cycle and averaged over runs, including retries.}
\label{tab:eval}
\small
\begin{tabular}{rrrrr}
\toprule
$n$ & merge levels & msgs sent/node & cycle time & verify time \\
\midrule
8   & 1 & 25.5 & 12\,ms  & 0.62\,ms \\
16  & 2 & 32.6 & 39\,ms  & 0.87\,ms \\
32  & 3 & 39.6 & 61\,ms  & 1.05\,ms \\
64  & 4 & 46.7 & 88\,ms  & 1.20\,ms \\
128 & 5 & 53.7 & 107\,ms & 1.37\,ms \\
256 & 6 & 60.6 & 148\,ms & 1.49\,ms \\
\bottomrule
\end{tabular}
\end{table}

In Table~\ref{tab:eval}, each merge level adds 6.9--7.1 messages per node, near
the $2c-1$ peer-and-ally contacts, and cycle time grows by 19--41\,ms per level.
Cycles cannot overlap, because each cycle's
assignment depends on the previous root, so input throughput rises only from 670
to 1{,}730 inputs per second as $n$ grows from 8 to 256; batching
(\S\ref{sec:req}) raises commitments per input.

\textbf{Simulated crash faults.} We simulated single cycles with random
fractions of silent nodes and random offset tables (Table~\ref{tab:sim}); at
$n=16{,}384$ and $c=16$, 8/16 cycles succeed at 5\%. The failures are isolated
ally-quorum misses (\S\ref{sec:protocol}): at each merge level $l\ge2$, a
level-$(l{-}1)$ group, one of $n/c$, with at least $b$ of its $c$ allies crashed
records a hole that the other groups computing the same value do not. Level 1
cannot split, since each level-0 stack is a single committee, so with $p$ given
by Eq.~\eqref{eq:pbad} for $f$ crashed nodes, a cycle succeeds with probability about $e^{-(L-1)(n/c)p}$. This matches
eight of the nine cells of Table~\ref{tab:sim} to within three runs and predicts
7.7/16 at $n=16{,}384$; it overestimates $c=64$ at 20\%, where holes compound
across levels, as an ally hole changes a group's value and blocks quorum at the
next level. Filling holes from windows (\S\ref{sec:protocol}) targets exactly
these misses.

\begin{table}[t]
\centering
\caption{Simulated cycles ($n=4{,}096$, 40 per cell) in which every live node
finishes on the same root, by committee size and share of crashed nodes, with
the first-order model's prediction in parentheses. --- marks cells not run.}
\label{tab:sim}
\small
\begin{tabular}{rccccc}
\toprule
$c$ & 2\% & 5\% & 10\% & 15\% & 20\% \\
\midrule
16 & 40 (40) & 37 (35) & 0 (0)   & ---     & --- \\
32 & ---     & 40 (40) & 36 (35) & 0 (1)   & --- \\
64 & ---     & ---     & 40 (40) & 38 (39) & 3 (9) \\
\bottomrule
\end{tabular}
\end{table}

The AWS runs used $c=4$, which tolerates one Byzantine member per committee. We consider
doing further tests with larger $n$, differing $c$ values as well as tests with different crashes, withholding and conflicting root strategies part of future work.

\section{Related Work}\label{sec:related}

\textbf{Aggregators.} A single aggregator receives $O(n)$ inputs while each
client sends $O(1)$ messages; the aggregator chooses which inputs enter the
Merkle root. Keyless signature infrastructure~\cite{buldas2013ksi} spreads this
aggregation over a hierarchy of servers, but a server can drop the inputs below
it without accountability, and witness cosigning~\cite{syta2016cosi} attests to the root
without sharing the choice of inputs. Cambium distributes the load and the
choice. Both an aggregator and Cambium need a
root-trust rule. Secure in-network aggregation~\cite{chan2006secure} also builds
a commitment tree inside the network and lets each contributor check its own
inclusion, but relies on a trusted querier.

\textbf{Timestamping, transparency, consensus.} Timestamping uses hash
chains or trees~\cite{haber1991timestamp}; Bitcoin
anchoring~\cite{opentimestamps} could supply a wall-clock reference for
Cambium's roots. Certificate Transparency maintains Merkle
logs~\cite{rfc9162}. PBFT and HotStuff totally order
requests~\cite{castro1999pbft,yin2019hotstuff}; interactive
consistency and asynchronous common subset deliver vectors or
sets~\cite{pease1980reaching,miller2016honeybadger}. Cambium computes an
aggregate instead. Algorand samples committees from prior randomness~\cite{gilad2017algorand}.
Cambium uses its own roots for re-seating and leaves proofs with participants.

\textbf{Local-first and accountable ledgers.} TrustChain keeps a
per-agent chain of interactions without global
consensus~\cite{otte2020trustchain}, and ConTrib's per-participant ledger detects
fraud through record exchange~\cite{devos2020contrib}. Cambium binds such
evidence to a shared root. It could anchor Secure Scuttlebutt's single-writer log
heads~\cite{kermarrec2020ssb}. Byzantine-tolerant version
control~\cite{nasrulin2021dcvc} and secure CRDTs~\cite{jannes2022secure}
replicate state; Cambium retains digests. Local economies may rely on
double-spend detection~\cite{lavoie2022local}; per-cycle inputs give range
evidence when senders disclose complete batches.

\section{Operating Shared Infrastructure}\label{sec:operating}

In a shared-data cooperative, contributors retain records and proofs that
auditors check against accepted roots. This leaves three operational questions.

\textbf{Membership and recovery.} $CR_0$ fixes the roster, and Sybil resistance
is delegated to whoever governs it. That suits institutional ledgers run as
commons~\cite{barisone2025swissledger} and closed commons that condition
replication on contribution~\cite{lavoie2023closed}, but a roster concentrates
authority unless it is accountable: admission and removal rules, authenticated
roster versions that keep old proofs verifiable, and a record of who authorised
each restart. The prototype stages signed future membership epochs, each
authorised by a threshold certificate from the current epoch, without breaking
the root chain; it changes $n$ between powers of two with $c$ and $CR_0$ fixed.

\textbf{Evidence availability.} Local retention avoids copying every record
everywhere, but applications then need exportable proof bundles and independent
retention, and verifiers must distinguish missing evidence from a negative claim.

\textbf{Participation and exclusion.} Someone pays for slow members: either
everyone waits for them, or timeouts turn their slow links into holes, so
low-bandwidth or intermittently connected members may repeatedly miss cycles,
and rebuild burdens may fall on them. An adversary can also target holes at
chosen committees (\S\ref{sec:protocol}). Because holes make missing participation
visible in the committed structure, a deployment can measure inclusion rate and
repeated holes per participant, and respond by tuning timeouts, subsidising
operation or delegating hosting. The operational test
for a community is whether participants can keep usable evidence, contest
exclusion and change operators without losing the ability to verify their
records.

\section{Conclusion}

Cambium lets participants build a shared commitment root per cycle, with no
aggregator or global ordering, while records and proofs stay with their holders. Number of messages per node, total build latency and proof sizes all grow logarithmically with the size of the network.
Its final roots agree, with high probability, whenever fewer than a third of
the nodes are Byzantine;
its liveness needs fewer than that, and untested hole filling is the main
engineering gap. Governance, proof availability and inclusion decide whether, and which, community it would best serve.

\bibliographystyle{ACM-Reference-Format}
\bibliography{refs}

\end{document}